\documentclass[aps,pra,twocolumn,notitlepage,10pt]{revtex4-1}
\usepackage{amsmath,amssymb,amsthm,mathtools}
\usepackage{graphicx,comment}
\usepackage{braket}
\newtheorem{theorem}{Theorem}

\usepackage{xcolor, color, soul}

\begin{document}
	\title{
    On the distinction between distinguishability of states and witness of non-Markovianity of dynamical maps}
	\author{Vijay Pathak}
	\email{vijay@ppisr.res.in}
	\author{R.Srikanth}
	\email{srik@ppisr.res.in}
	\affiliation{%
		Theoretical Physics Department, Poornaprajna Institute of Scientific Research Bengaluru, India 562 164
	}%
	\date{\today}
	
	\begin{abstract}
Non-P-divisibility is the strongest divisibility-based notion of quantum non-Markovianity. The generalized trace distance (GTD) based criterion is known to be an optimal witness of non-P-divisibility of dynamical maps, in the sense that a given map is non-P-divisible if and only if there exists a pair of states that demonstrates increased distinguishability in the GTD sense. This observation forms the basis for associating an information backflow with this type of non-Markovianity. Here we argue that in contrast to the map-level witnessing of non-Markovianity via divisibility, the association of information flow with divisibility must be applicable to individual states or state pairs (in the trace-distance context). In the context of qubit dynamics, we show that this association is generally neither tight nor faithful. We demonstrate this by means of counter-examples: (a) a pair of states whose distinguishability manifestly increases, but the GTD criterion fails to indicate this. (b) manifestly indistinguishable states that are indicated to be GTD distinguishable. In other words, we point out a subtle distinction between indicating state-specific behavior in terms of information backflow or distinguishability and map-level witnessing of non-Markovianity based on the generalized trace distance (GTD). Furthermore, we demonstrate that for qubit unital dynamics, the GTD-based measure provides no advantage over the standard trace distance measure in witnessing non-Markovianity. We determine the class of qubit non-unital channels where the standard trace distance measure is insufficient and the generalized measure is necessary.  
	\end{abstract}
	
	\maketitle
	
	\section{Introduction}
The isolation of a quantum system from its environment is an ideal condition, rendering open systems a significant object of study \cite{BreuerBook}. Such investigations are involved in achieving quantum control, focusing on the analysis of complete positive channels (CP) \cite{CHOI1975285, PhysRev.121.920}. CP maps distinguish themselves from positive (P) maps \cite{Choi_1972} by remaining unaffected by including ancillary reference systems that do not participate during evolution. Any uncorrelated system and the environment with a combined Hamiltonian evolution guarantees the CP evolution. Although certain quantum-correlated \cite{PhysRevLett.102.100402, doi:10.1142/S1230161223500117} states may induce CP dynamics, such occurrences are not universal across all unitary dynamics. The correlations generally lead to various dynamical maps, not just CP maps \cite{ PhysRevA.70.052110, PhysRevLett.73.1060, PhysRevA.64.062106, Colla_2022, dijkstra2012non}.

From start to finish, the channel goes through different temporal dynamics depending on various possible environments and their interaction with the system. This time evolution determines different processes based on these interactions. Should the environment be sufficiently large or the interaction remain sufficiently weak, it is reasonable to assume that the environment does not vary significantly alongside the system and that the correlations are not strong enough to be considered significant. Under such circumstances, the time evolution is laid out by the Gorini-Kossakowski-Sudarshan-Lindblad (GKSL) \cite{10.1063/1.522979, Lindblad1976} time-local equation, which follows the semi-group property characteristic of Markovian evolution. As quantum technology advances, allowing for experimentally greater control of the system's dynamics, the above assumptions can be relaxed \cite{white2020demonstration, bernardes2015experimental, PhysRevA.109.042419, liu2013photonic}. In that situation, semi-group dynamics will be replaced by general divisible dynamics, provided the dynamics is invertible. 

The divisibility of a dynamics is defined by:
\begin{equation}
\label{Eq:divisible}
\mathcal{E}_{(t,0)} = \Lambda_{(t,\tau)}\circ\mathcal{E}_{(\tau,0)}
\end{equation}
where, $\Lambda$ is the intermediate map taking the evolution from time $\tau$ to $t$.

Semi-group dynamics is a special case of divisible dynamics where $\Lambda_{(t,\tau)} = \mathcal{E}_{t-\tau}$. In general, even if the dynamics $\mathcal{E}_{(t,0)}$ from the initial time to any final time $t$ is CP, the intermediate map may or may not be CP, depending on the system-environment interactions. Correspondingly, the divisible dynamics is categorized into CP-divisible and non-CP-divisible. Non-CP-divisible dynamics can be P-divisible and non-P-divisible, depending on whether the (non-CP) intermediate map is P or non-P. Non-CP-divisible dynamics is considered non-Markovian based on divisibility \cite{PhysRevLett.105.050403}, while non-P-divisible dynamics is considered a stronger manifestation of non-Markovianity.

Breuer-Laine-Piilo \cite{PhysRevA.81.062115, PhysRevLett.103.210401} introduced an information-theoretic measure called BLP to understand the memory effect \cite{RevModPhys.88.021002, liu2013photonic}. It is based on the distinguishability of two quantum states, as indicated by the trace distance measure. The CP and P processes weaken distinguishability, suggesting a loss of information to the environment. Since BLP does not respond to the translation part of the dynamics \cite{liu2013nonunital}, a generalized trace distance measure (GBLP) was introduced \cite{PhysRevA.83.052128}. Here, the increment of the generalized trace distance for finite time is identified with non-Markovianity. In an attempt to unify the non-Markovianity based on trace distance and divisibility, a theorem in Ref. \cite{PhysRevLett.121.080407} shows that P-divisibility is connected to a trace distance condition, specifically the GBLP condition \cite{PhysRevA.83.052128, PhysRevA.92.042108} rather than the BLP condition, i.e., an invertible dynamical map is P-divisible if and only if a monotonic non-increase of GTD distinguishability is observed for any pair of states of the system in question. The GBLP condition thus serves as an optimal witness of the non-P-divisibility of a map. In related work, Ref. \cite{bylicka2017constructive} shows that for an invertible dynamical map, CP-divisibility is equivalent to a monotonic non-increase of distinguishability for two equiprobable states of the system together with an ancilla, giving an explicit method to construct a pair of such states in the case of a non-CP-divisible map.

These considerations suggest that the GBLP criterion should be the right indicator of the information flow \cite{PhysRevLett.121.080407, PhysRevA.83.052128, PhysRevA.92.042108}. It is worth stressing that the optimality of the GBLP condition as a witness of non-P-divisibility applies to the map as a whole rather than a given pair of states. Consider the non-P-divisible qubit dephasing map given by $\rho \xrightarrow{} \mathcal{E}_t(\rho) = \frac{1+\cos^2(\omega t)}{2}\rho + \frac{1-\cos^2(\omega t)}{2}Z\rho Z$. The pair of states $\ket{0}, \ket{1}$ remains manifestly invariant and would be unsuitable to witness the non-P-divisibility of the map. On the other hand, any pair of distinct states situated at the same polar angle will be swept periodically to an identical point on the $z$-axis before moving apart. Thus, the existence of pair(s) of states whose distinguishability doesn't evolve doesn't undermine the non-P-divisibility status of a map.

Even so, we expect that the concept of distinguishability or information flow must be applicable even at the microscopic level, i.e., at the level of pairs of individual evolving states. Any pair of states whose distinguishability manifestly evolves must be witnessed as such by the GTD criterion and vice versa. In this context, a tight witness is one that doesn't admit false negatives, i.e., all positive cases of outflow or inflow are identified as such and not flagged as no-flow. On the other hand, a faithful witness doesn't admit false positives (only positive cases are indicated as such). We may accept a non-tight witness, but an unfaithful witness would merit caution.

This raises the question of whether the GTD criterion can serve as a faithful witness to distinguishability or information flow when considering individual pairs of states. Here, we will find that this is not the case: it is non-tight in that there can be a pair of states with evolving distinguishability, but the GTD criterion does not indicate this. More worryingly, it is not faithful in that a pair of states can be manifestly undistinguishable but shown to be GTD distinguishable. In other words, we will find that the GTD criterion, based on which GBLP is neither a tight nor even a faithful witness of information backflow. 

The paper is written in the following way: In \ref{Sec: Measure information-flow}, the affine transformation of a qubit is defined, and the BLP and GBLP measures are introduced to characterise non-Markovianity. An example is shown for which BLP fails, but GBLP works. In Sec. \ref{Sec: GBLP}, the general class of dynamics where GBLP is necessary is determined. In Sec. \ref{Sec: information-flow not GBLP}, it is shown that GBLP is not the right indicator, in general, to capture the information flow. Finally, we conclude with discussions and conclusions in Sec. \ref{sec:disc}.

\section{\label{Sec: Measure information-flow}Non-Markovian measures for information backflow}

An affine transformation can represent the dynamics of a qubit interacting with an arbitrary quantum environment. Let us write the state in Pauli basis, then $\rho = \dfrac{1}{2}(\mathbb I + \Vec{r} \cdot \Vec{\sigma})$ with $\lvert\vec{r}\rvert\leq1$, here $\sigma's$ are Pauli matrices and $\Vec{r}$ is a Bloch vector representing the state in the Bloch sphere. The Bloch vector transformation is given by 
\begin{equation}
\label{eq:affine}
 \Vec{r}' = T \Vec{r} + \Vec{c}
 \end{equation}
The transformation has two components: (a) a $3\times 3$ $T$ matrix representing the rotation, contraction and expansion, and (b) a $3\times 1$ $\Vec{c}$ vector representing the translation. If $\Vec{c}=\Vec{0}$, then the dynamics is called unital dynamics, which leaves the completely mixed state, the centre of the Bloch sphere, invariant. The non-zero translation vectors represent non-unital dynamics. 

 \subsection{BLP measure}
Breuer-Laine-Piilo
\cite{PhysRevA.81.062115, PhysRevLett.103.210401} defined a measure that connects non-Markovianity with the information back-flow. This measure is important because it highlights the memory dependence on dynamics, which was not possible through divisibility-based non-Markovianity. It is not necessary to know the dynamical map to experimentally witness the non-Markovian behaviour of the dynamics. For defining the measure, they considered the trace distance defined as 
\begin{equation}
\label{BLP: def}
D = \dfrac{1}{2}\lVert \rho_1 - \rho_2 \rVert_1
\end{equation}
here $\lVert A \rVert_1 = \mathrm{Tr} \lvert A \rvert$, where modulus is defined as $\lvert A \rvert = \sqrt{A^\dagger A}$. This distance was shown to decide the distinguishability of two states prepared with equal probability:
\begin{equation*}
    P_{\rm dis} = \dfrac{1+D}{2}
\end{equation*}
here, $P_{\rm dis}$ is the probability to distinguish two states.
Unitary dynamics keeps this distance unchanged for all initial pairs; hence, distinguishability is unchanged under unitary dynamics. This is consistent with the fact that the system does not interact with any other quantum system under unitary dynamics, so there is no exchange of information. In the event of interactions based on distinguishability, there are two categories of dynamics: (a) Markovian dynamics: the distance does not increase for any initial pair for any interval of time. In this case, there is never an increase in distinguishability. (b) Non-Markovian dynamics: the distance increases at least for a few initial pairs for some time interval. In this case, the distinguishability may increase for an interval of time. 
 
CP-divisible dynamics is CP for all intermediate initial to intermediate final times; this will keep contracting for all intermediate times since CP maps contract. P-divisible dynamics is CP from the initial to the final time, but for some intermediate times, it can be just P, not CP; it will also continue to contract because P-maps also contract. Because of this, CP-divisible and P-divisible dynamics are considered Markovian according to this measure \cite{PhysRevLett.121.080407}.

\subsection{Generalized measure\label{GM}}
As we can see from Eq.~(\ref{BLP: def}), this distance is unaffected by the translation vector $\vec{c}(t)$. So, the above measure will not capture any information flow associated with that. 
 
Thus, the BLP definition does not address the translation aspect of the dynamics. If two maps have the same $T$ part with different $\vec{c}$, then BLP can not differentiate them. To overcome this, the GBLP measure, discussed below, was introduced.
The Helstorm Matrix 
\begin{equation}
\Delta_p(\rho_1,\rho_2) = p\rho_1 - q\rho_2
\end{equation}
with $p+q=1$ was considered for the distance to overcome this issue, as it is sensitive to translations \cite{PhysRevA.83.052128}. The generalized distance is defined as
\begin{equation}
    D = \lVert \Delta_p(\rho_1,\rho_2) \rVert_1,
\end{equation}
which is also a distance measure between two states, but prepared with biased probabilities. This is also shown by connecting the distinguishability of two states with the generalized distance \cite{PhysRevA.83.052128, PhysRevA.92.042108}. Unlike the BLP measure, P-divisibility is known to be connected with this generalized measure with the if and only if condition. 

As a simple example where BLP fails but GBLP works, consider the non-P-divisible generalized amplitude damping noise given by the following affine map $\mathcal{E}_t(\vec{r}) = T(t)\vec{r} + \vec{c}(t)$, where
\begin{align}\label{Eq: GAD}
    T(t) &= \begin{pmatrix}
   \sqrt{\eta(t)}      & 0 & 0 \\
     0    & \sqrt{\eta(t)}      & 0 \\
     0 & 0 & \eta(t) 
    \end{pmatrix}
    \nonumber\\
    \vec{c}(t) &= \begin{pmatrix}
       0 \\ 
       0 \\
         (2s(t)-1)(1-\eta(t))
    \end{pmatrix}
\end{align}
where, $\eta(t) = e^{-\gamma t}$ for $\gamma > 0$ and $s(t)=\cos^2{ft}$ for all $f$. Consider the initial states $\vec{r}_1 \equiv (x_1,y_1,z_1)^{\dagger}$ and $\vec{r}_2 \equiv (x_2,y_2,z_2)^{\dagger}$. We have the trace distance between their evolved versions to be
\begin{align}
    \frac{1}{2}\textrm{Tr}[\mathcal{E}_t(\rho_1)-\mathcal{E}_t(\rho_2)] &= \frac{1}{2}|T(t)\vec{r}_1-T(t)\vec{r}_2|
    \nonumber\\
    &=\frac{\sqrt{e^{-\gamma t}(\delta^2 x + \delta^2 y) + e^{-2\gamma t}\delta^2 z}}{2},
\end{align}
where $\delta x\equiv (x_2 - x_1)$, $\delta y\equiv (y_2 - y_1)$, and $\delta z\equiv (z_2 - z_1)$. The above equation implies a falling trace distance for all initial pairs, $\gamma$ and $f$, hence BLP Markovianity. With GBLP, the distance
\begin{align}
    \textrm{Tr}[p\mathcal{E}_t(\rho_1)-q\mathcal{E}_t(\rho_2)] &\neq |pT(t)\vec{r}_1-qT(t)\vec{r}_2 + (p-q)\vec{c}(t)|
\end{align} 
As a simple illustration, for $\vec{r}_1 \equiv (1,0,0)^{\dagger}$, $\vec{r}_2 \equiv (0,1,0)^{\dagger}$, $p=0.25$, $f=4$, and $\gamma=0.1$, the required generalized trace distance is plotted in Fig~(\ref{fig:example}). The distance is oscillating with time. 
\begin{figure}[ht]
    \centering
    \includegraphics[width=8.6cm, height=5.5cm]{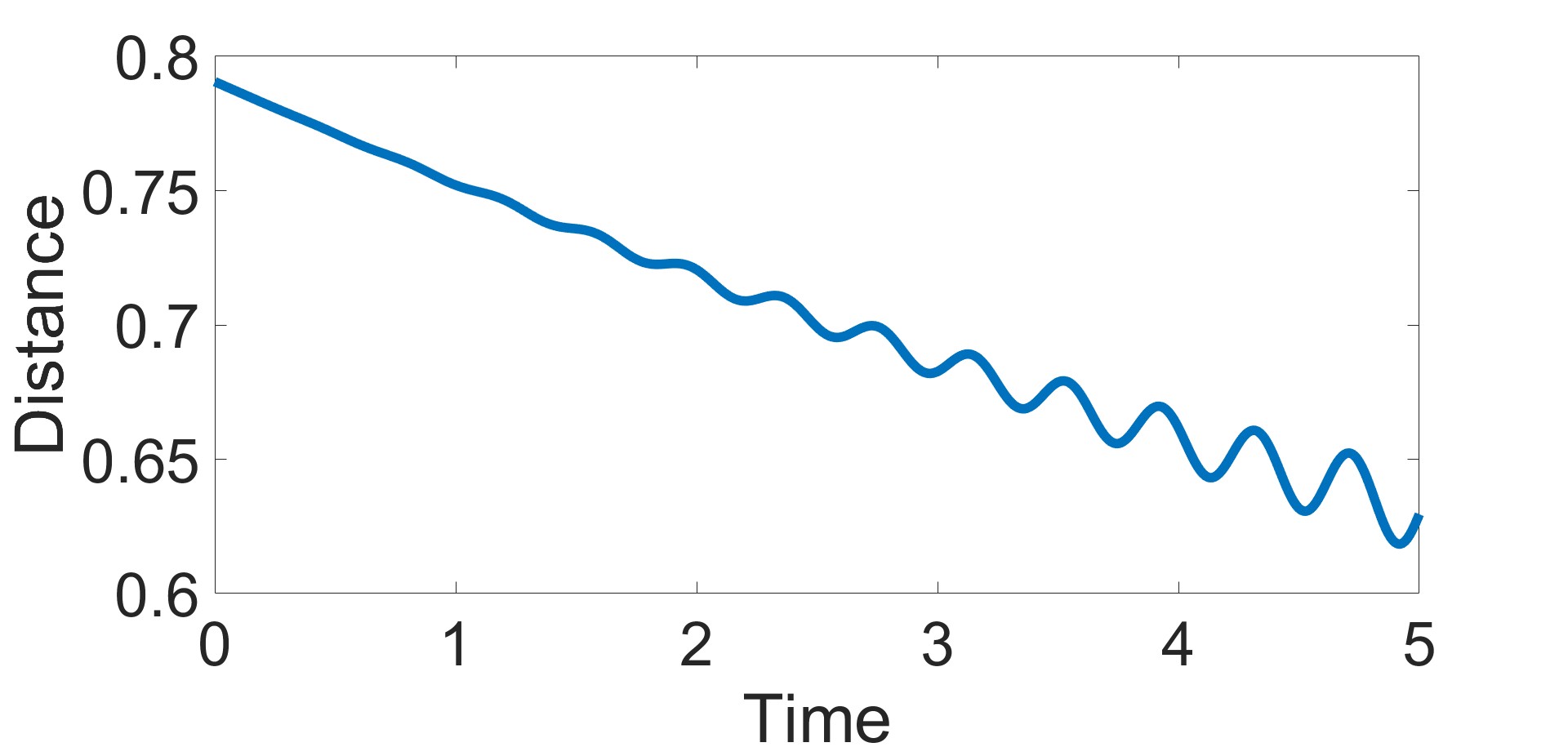}
    \caption{generalized distance between states $\dfrac{1}{2}(\mathbb I + \sigma_z)$ and $\dfrac{1}{2}(\mathbb I + \sigma_x)$ for $p=0.25$ is plotted as a function of time for generalized amplitude damping noise (See Eq.~\ref{Eq: GAD}). Noise parameters $\gamma$ and $f$ are 0.1 and 4, respectively. Distance and time are in arbitrary units.}
    \label{fig:example}
\end{figure}
Thus, according to the theorem \cite{PhysRevLett.121.080407}, which states that non-P-visible dynamics will necessarily have non-decreasing generalized trace distance for some initial pair for a certain interval of time, this channel is non-P-divisible and hence non-CP-divisible. As such, it illustrates the usefulness of the GBLP measure where BLP fails. 

\section{\label{Sec: GBLP}Unitality, divisibility and distance}
The above represents one example where the two measures differ; however, we aim to identify the general classes in which they differ. We are checking the distance evolution for both BLP and GBLP measures for CP, P and non-P maps. Once that is understood, the connection with the whole dynamics is straightforward because, for divisible dynamics, the whole dynamics is nothing but the composition of the intermediate maps. The GBLP measure is calculated for two initial states as $\rho_{i} \equiv (x_{i}, y_{i}, z_{i})$, where $i=1,2$ and probability $p$. The eigenvalues of $\Delta^{\dagger}\Delta$ are 
\begin{equation*}
\dfrac{1}{4}\bigg(\sqrt{(p x_1 - q x_2)^2 + (p y_1 - q y_2)^2 + (p z_1 - q z_2)^2} \pm \lvert p-q \rvert\bigg)^2
\end{equation*}
To calculate the distance, we add the square root of the eigenvalues, as the trace does not depend on the basis. The generalized distance is
\begin{equation*}
    D = \dfrac{1}{2}\lvert \lvert \Vec{w} \rvert + \lvert p-q \rvert\rvert + \lvert \vert \Vec{w} \rvert - \lvert p-q \rvert\rvert
\end{equation*}
Here, $\Vec{w} = p\Vec{r_1}-q\Vec{r_2}$. After expanding the above equation,
\begin{equation}   
D_t = 
\begin{array}{ll}
  \lvert \Vec{w}(t) \rvert, & \lvert \Vec{w}(t) \rvert \geq \lvert p-q \rvert \\
  \lvert p-q \rvert, &  \lvert \Vec{w}(t) \rvert < \lvert p-q \rvert \\
\end{array} 
\label{eq:Dt}
\end{equation}
If $p=q$, the GBLP distance reverts to the BLP distance.

\begin{theorem}
The GBLP and BLP conditions are equivalent for qubit unital dynamics.
\label{thm:BLP}
\end{theorem} 
\begin{proof}
The translation vector vanishes for unital dynamics, and the dynamics is characterised by rotation, compression or expansion. Eq. (\ref{eq:Dt}) reduces to:
\begin{equation}
D_t = 
\begin{array}{ll}
  \lvert T(t)(p\Vec{r}_{1} - q\Vec{r}_{2}) \rvert, & \lvert \Vec{w}(t) \rvert \geq \lvert p-q \rvert \\
  \lvert p-q \rvert, &  \lvert \Vec{w}(t) \rvert < \lvert p-q \rvert \\
\end{array} 
\label{eq:BPLD}
\end{equation}  

We note the following:
\begin{enumerate}
\item If the channel is GBLP NM, then $D_t(p\vec{r}_1, q\vec{r}_2)>D_{t-\delta}(p\vec{r}_1, q\vec{r}_2)$ for some $p$, $\Vec{r}_1$, $\Vec{r}_2$, $\delta$ and $t$. Since $p, q\leq1$, $p\Vec{r}_1$ and $q\Vec{r}_2$ will also be valid vectors belonging to the Bloch sphere. Hence same distance can be written in terms of $\vec{r'}_1$ and $\vec{r'}_2$ where $\Vec{r'}_1=p\Vec{r}_1$ and $\Vec{r'}_2=q\Vec{r}_2$. The condition $D_t(p\vec{r}_1, q\vec{r}_2)>D_{t-\delta}(p\vec{r}_1, q\vec{r}_2)$ will become $D_t(\vec{r'}_1, \vec{r'}_2)>D_{t-\delta}(\vec{r'}_1, \vec{r'}_2)$ which is the condition for BLP-NM. 

\item If the channel is BLP NM, then $D_t(\vec{r}_1, \vec{r}_2)>D_{t-\delta}(\vec{r}_1, \vec{r}_2)$ for some $\Vec{r}_1$, $\Vec{r}_2$, $\delta$ and $t$. This condition gives one pair that belongs to GBLP with $p=\dfrac{1}{2}=q$. That makes GBLP also NM.
\end{enumerate}
\end{proof}

As a corollary of Theorem \ref{thm:BLP}, there is no unital non-P-divisible BLP Markovian dynamics. We now consider the question of what are the channels for which the BLP and GBLP criteria are nonequivalent.
 
\begin{theorem}\label{theorem2}
In the case of qubit dynamics, BLP and GBLP differ only for non-unital non-P-divisible BLP dynamics. 
\end{theorem}
\begin{proof}
    Our earlier example demonstrated that the BLP and GBLP conditions are not equivalent. From Theorem \ref{thm:BLP}, it follows that, nevertheless, they are equivalent for unital noise. Thus, their non-equivalence can hold only for non-unital noise. Here, by Ref.~\cite{PhysRevLett.121.080407}, we know that P-divisibility is equivalent to the GBLP Markovianity. Thus, the non-equivalence of the BLP and GBLP conditions can arise only from non-unital non-P-divisible noise. If this noise is BLP non-Markovian, then the two criteria will agree.
\end{proof}

\section{\label{Sec: information-flow not GBLP}Information flow and distinguishability}

The idea of information flow was introduced as an interpretation of distinguishability in the BLP definition. If the distinguishability is invariant between all pairs, then this indicates no information exchange, suggesting that the dynamics is unitary. On the other hand, if the distance monotonically decreases, the distinguishability decreases, suggesting that the information should be transferred away from the system.

This intuition was justified by defining the internal information ($I_{int}$) between system states using BLP and GBLP distances, which can be accessed through the measurement of the system alone \cite{RevModPhys.88.021002}. The external information ($I_{ext}$), which cannot be accessed alone from the system, lies in correlation and/or within environment states and is defined by subtracting from total information ($I_{\rm total}$), which is from the combined state of the system-environment. So 
\begin{equation}
\label{Eq: InformationFlow}
I_{\rm int}^t(\rho_S^1,\rho_S^2)+I_{\rm ext}^t=I_{\rm total}^t(\rho^1,\rho^2),
\end{equation}
where $I_{\rm ext}=I_{\rm total}-I_{\rm int}$ and $\rho_S$, $\rho$ are system and combined system-environment state respectively. Since total information is constant owing to the unitarity of the dynamics, $I_{\rm total}^t=I_{\rm total}^0$. All information forward-flow or back-flow arguments are justified by Eq.~(\ref{Eq: InformationFlow}). This suggests that information flow should be associated with individual pairs of states, unlike properties such as P-divisibility or CP-divisibility, which characterize a map as a whole. 

While GTD is thus known as a valid witness of non-Markovianity at the level of maps, here we wish to scrutinize the tightness and faithfulness of the association of distinguishability or information flow with GTD at the level of states or pairs of states.
Below, we will provide specific instances to argue that it is neither a tight nor a faithful indicator.

\subsection{False negative test of distinguishability}
We now demonstrate that the GTD criterion is not a tight indicator of distinguishability when applied to specific pairs of states.
The GBLP definition does not follow distinguishability and information flow for all pairs \textit{at all times}. There is a possibility that states can evolve and come closer or go far, but the information flow is indicated to be zero.

$\textbf{Examples}$: 
Let us consider the unital dynamics with 
\begin{equation}\label{Eq: isometric}
    T(t) = \left(\begin{array}{ccc}
   \exp^{-\gamma t} & 0 & 0  \\
   0 & \exp^{-\gamma t} & 0  \\
   0 & 0 & \exp^{-\gamma t}
\end{array}\right)
\end{equation}
According to Eq.~(\ref{Eq: InformationFlow}), the BLP definition suggests the information outflow for all pairs for all time. In light of Eq.~(\ref{eq:BPLD}), the generalized distance is constant for Bloch radius $r\leq0.5$ in the considered example for $\gamma=0.1$ and $p=0.25$. Thus, GBLP will show no outflow or distinguishability change for certain pairs for corresponding probabilities, as shown in Fig~(\ref{fig:isometric_plot}). This contradicts our expectation that the absence of information flow should only be true for unitary dynamics or pairs where their subspaces do not evolve. This suggests that the generalized definition of distance does not capture the information flow for all pairs as required by Eq.~\ref{Eq: InformationFlow}. 
\begin{figure}[ht]
    \centering
    \includegraphics[width=8.6cm, height=5.5cm]{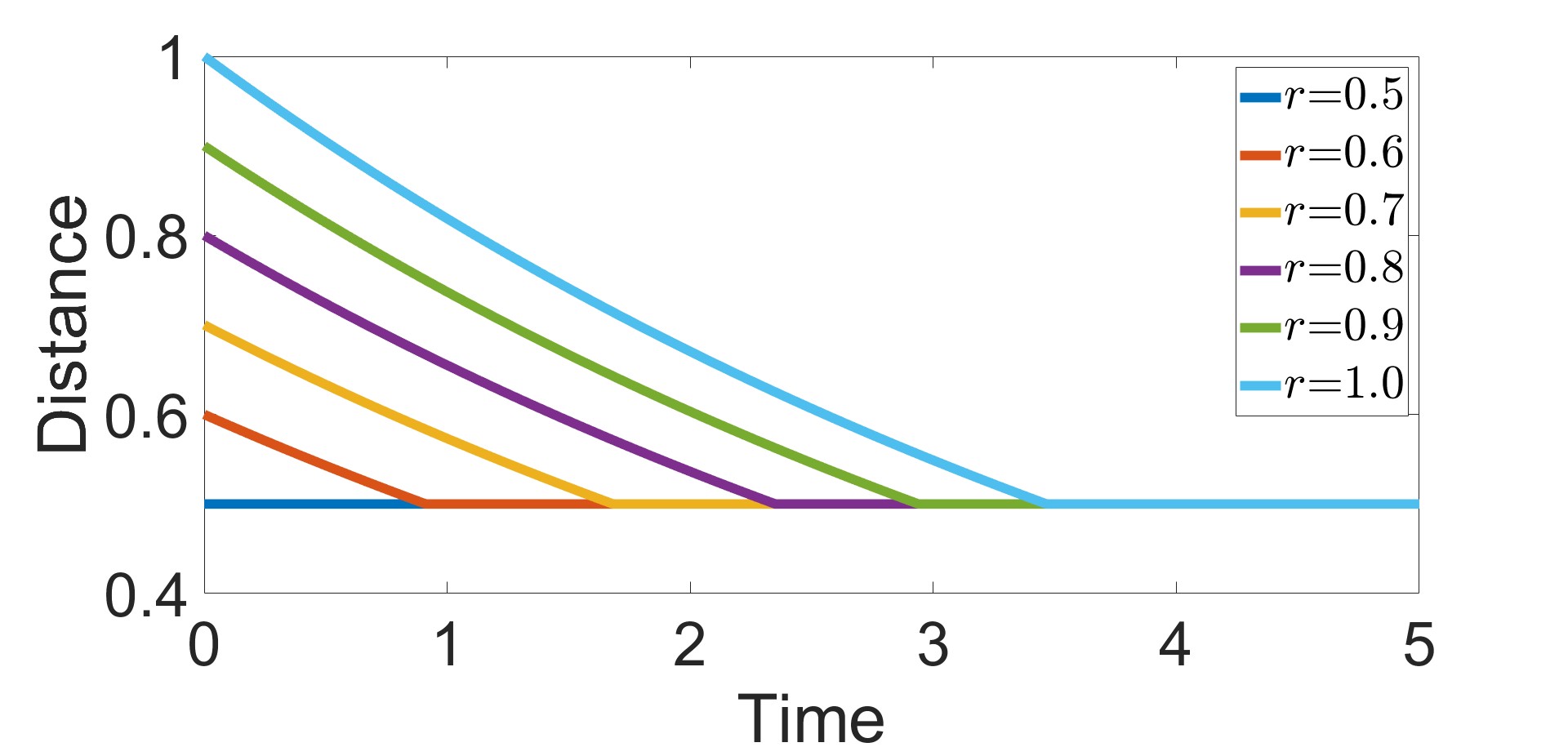}
    \caption{Generalized distance is plotted as a function of time for different Bloch radii $r$ with two diametrically opposite points in the Bloch sphere ($\theta=\pi/7$, $\phi=\pi/5$, $r$ is varied) for an isometrically decaying dynamical map (Eq.~\ref{Eq: isometric}). The decay parameter $\gamma$ and the probability of preparing states $p$ are 0.1 and 0.25, respectively. Distance and time are in arbitrary units.}
    \label{fig:isometric_plot}
\end{figure}

\begin{figure}[ht]
    \centering
    \includegraphics[width=8.6cm, height=5.5cm]{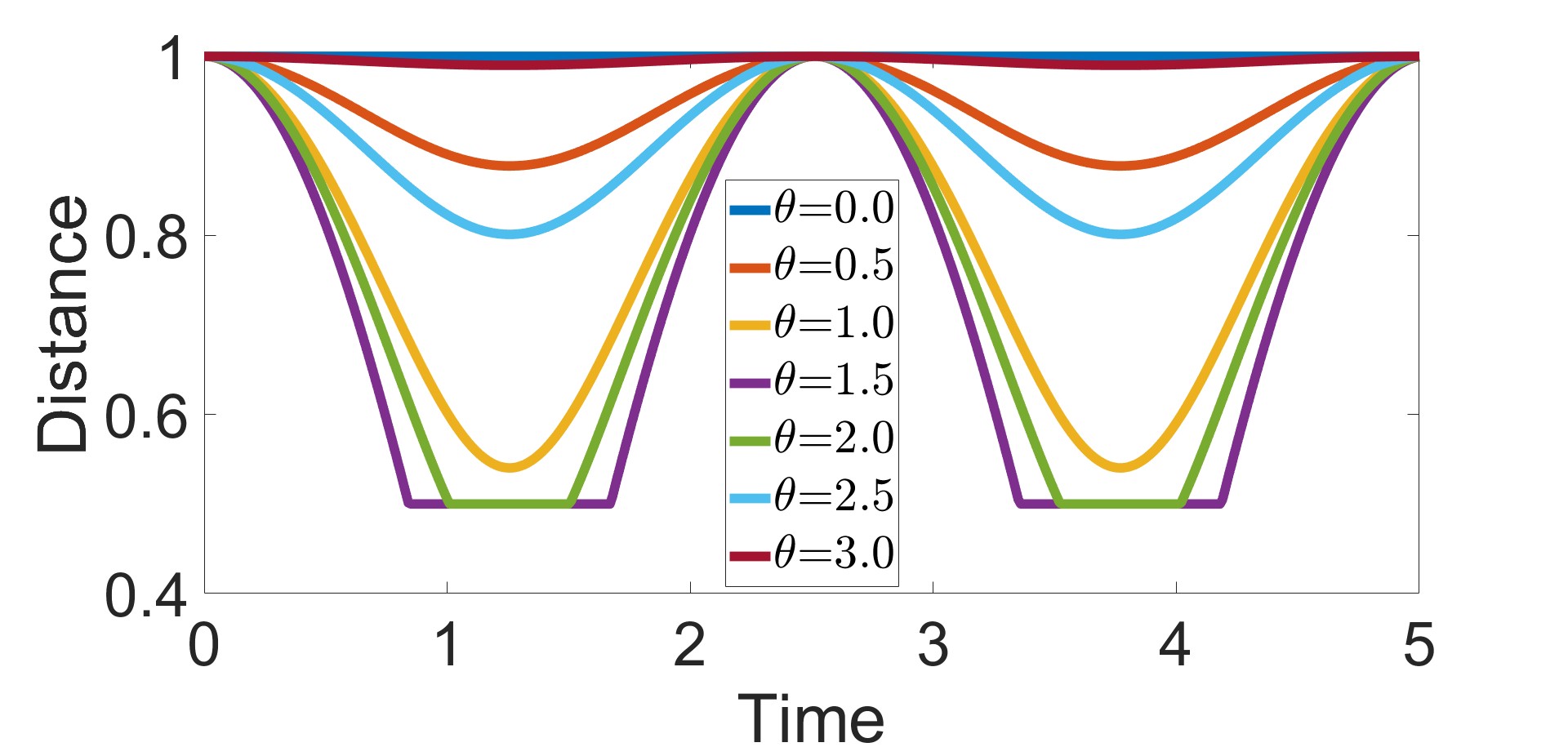}
    \caption{Generalized distance is plotted as a function of time for different pairs of diametrically opposite points on the surface of the Bloch sphere ($r=1$, $\phi=0.5$, $\theta$ is varied) for a spin interacting with another spin (Eq.~\ref{Eq: Spin}). The dynamical parameter $\omega$ and the probability of preparing states $p$ are 1.25 and 0.25, respectively.  Distance and time are in arbitrary units.}
    \label{fig:spinchain_plot}
\end{figure}
This issue is not restricted to BLP Markovian dynamics. For BLP non-Markovian dynamics, the domain exists, showing no outflow or inflow, even though states manifestly evolved during that period. To illustrate this, two interacting spins are considered, with one spin serving as the system and the other as the environment. The dynamical map is given by:
\begin{equation}\label{Eq: Spin}
T(t) = \left(\begin{array}{ccc}
   \cos{\omega t} & 0 & 0  \\
   0 & \cos{\omega t} & 0  \\
   0 & 0 & 1  \\
\end{array}\right) 
\end{equation}
This non-invertible process is trivially non-Markovian and exhibits evident information backflow. Yet, for certain pairs of states, it can evolve their distance from a regime of $\lvert \Vec{w}(t) \rvert \geq \lvert p-q \rvert$ to one of $\lvert \Vec{w}(t) \rvert < \lvert p-q \rvert$, which, in light of Eq. (\ref{eq:Dt}), can cause the GBLP measure to flatten out. The non-changing phase is shown in Fig~(\ref{fig:spinchain_plot}) for $\omega=1.25$ and $p=0.25$, which shows no information flow from or towards the system. We can observe in this example that pairs along the z-axis do not evolve, which is consistent for both BLP and GBLP measures.

These examples suggest that GBLP is a measure of contracting maps related to P-divisibility. This does not capture the idea of information flow for all cases at all times. Even in our examples, we see that there are other time intervals where GBLP faithfully captures information outflow or backflow (the non-flat regimes of in Figure \ref{fig:spinchain_plot}), and other states where GBLP faithfully captures information outflow or backflow at all times. We stress that our result doesn't undermine the optimality of GBLP as a witness of the non-P-divisibility of a map as a whole, but rather highlights that it isn't a faithful indicator of information backflow in the context of individual pairs of states.

\subsection{False positive test of distinguishability}
More worryingly, we now demonstrate that the GTD criterion can be unfaithful in the sense we discussed when applied to specific pairs of states. To give an example of a case where the GBLP condition can be misleading, consider the composition of two channels as $\mathcal{E}_{t,0}=\mathcal{B}_{t,\tau}\circ\mathcal{A}_{\tau,0}$ where $\mathcal{A}$ acting till time $\tau$ starting from zero is a CPTP depolarising channel $\mathcal{A}(\rho) = (1-p)\rho + p\frac{\mathbb{I}}{2}$ corresponds to the unital channel with matrix form given by $\mathcal{A}(\tau) = \left(\begin{array}{cccc}
   1 & 0 & 0 & 0 \\
   0 & 1-p(\tau) & 0 & 0  \\
   0 & 0 & 1-p(\tau) & 0  \\
   0 & 0 & 0 & 1-p(\tau)  \\
\end{array}\right) $ and $\mathcal{B}$ is the shift channel given by $\mathcal{B}(t,\tau) = \left(\begin{array}{cccc}
   1 & 0 & 0 & 0 \\
   0 & 0 & 0 & 0  \\
   0 & 0 & 0 & 0  \\
   c & 0 & 0 & 0  \\
\end{array}\right) $. Here, the depolarising probability is $p(0)=0$ and $p(\tau)=1$. The shift is assumed to be sufficiently small not to leave the Bloch ball. More generally, suppose the shift dynamics is given by an increase in Bloch vector length given by $c(1-|\vec{r}|)$, which ensures that positivity is preserved. At time $\tau$, all initial states will go to a completely mixed state. The GBLP measure evaluates the distance to zero at time $\tau$, and after the shift channel, the distance will be $\lvert(p-q)\rvert$ at time $t$ irrespective of the value of $c$. Though the states are identical at time $\tau$, this non-zero distance after the shift gives the appearance of distinguishing them because it has a label of biased probabilities, which reflects the preparation information rather than information back-flow from the environment.

The divisible dynamics mentioned in Theorem \ref{theorem2}, where GBLP is necessary, will necessarily give false positives. This can be specifically seen in the example shown in \ref{GM}.

\section{Discussion and Conclusion \label{sec:disc}}
 
The generalized BLP measure was needed to characterise the non-Markovianity of maps because BLP did not capture aspects of the non-unital noise. Specifically, it was shown that GBLP has a direct connection with the P-divisibility of maps. We have demonstrated that GBLP and BLP are equivalent witnesses of non-Markovianity for unital dynamics. For non-unital cases where the Bloch sphere is contracting, BLP does not help decide the non-Markovianity of maps. In those cases, GBLP shows that distance can increase solely with time due to the translation term associated with information backflow. 

While GBLP is a valid witness of non-Markovianity of dynamical maps (via optimization), here we argue that the association of information flow or distinguishability with divisibility is neither tight nor faithful when applied to individual states or pairs of states (in the trace-distance context). We demonstrate this using specific counter-examples: (a) a pair of states whose distinguishability manifestly increases, but the GTD criterion fails to indicate this. (b) manifestly indistinguishable states that are indicated to be GTD distinguishable.

Here, the GBLP case for a qubit, as an example, is considered to highlight this difference; other specific definitions and higher dimensions can be explored in the future. What kind of information non-unital dynamics possess within their translation part should also be explored to understand the information aspect of non-Markovianity for any dynamics.

However, we can consider ``trivial'' extensions to higher dimensions to demonstrate the relevance of our results for these systems. An example here would be noise action on a qutrit (3-dimensional quantum system), of the form $\mathcal{E}_2 \oplus \mathbb{I}_1$, i.e., nontrivial dynamics in a qubit subspace of a qutrit and identity in the other 1-dimensional level. The partitioned qutrit density matrix can be conveniently represented as:
\begin{equation}
\rho = \begin{pmatrix}
\frac{1}{2}(w + r_z) & \frac{1}{2}(r_x - i r_y) & \rho_{02} \\
\frac{1}{2}(r_x + i r_y) & \frac{1}{2}(w - r_z) & \rho_{12} \\
\rho_{02}^* & \rho_{12}^* & 1-w
\end{pmatrix},
\end{equation}
where the coherences $\rho_{02}, \rho_{12}$ and population $\rho_{22}=1-w$ remain invariant. Within the qubit subspace, the rescaled subspace dynamics has the usual form
\begin{equation}
\vec{r'} = T\vec{r} + w\vec{c},
\end{equation}
where the displacement vector $\vec{c}$ is scaled by $w$ because the ``fixed point'' of the dynamics (like the ground state in amplitude damping) is scaled by how much of the system is actually ``living'' in that subspace. On the other hand, note that by definition, $r_i = \text{Tr}(\rho \sigma_i)$. If only a fraction $w$ of the population is in the qubit subspace, the maximum value any $r_i$ can take is $w$. Thus, $\vec{r}$ is naturally ``pre-scaled'' by the state's geometry. 

Therefore, our above analysis of qubit dynamics carries over to this specific case of qutrit dynamics, except that instead of a unit Bloch sphere ($|\vec{r}| \leq 1$), the dynamics occurs within a rescaled Bloch sphere of radius $w$. One can similarly show that qubit dynamics embedded in a composite system or an encoded system allows us to retain our above conclusions. More general noise in higher dimensions remains a topic for future study.

\section*{Acknowledgments}

VP acknowledges the financial assistance of the Anusandhan National Research Foundation (ANRF) through grant CRG/2022/008345.	

\bibliography{main.bib}
\end{document}